\documentclass[10pt,reqno]{amsart}
     \makeatletter
     \def\section{\@startsection{section}{1}%
     \z@{.7\linespacing\@plus\linespacing}{.5\linespacing}%
     {\bfseries
     \centering
     }}
     \def\@secnumfont{\bfseries}
     \makeatother
\newtheorem{theorem}{Theorem}[section]

\newtheorem{proposition}[theorem]{Proposition}

\theoremstyle{definition}
\newtheorem{definition}[theorem]{Definition}

\theoremstyle{remark}
\newtheorem{remark}[theorem]{Remark}
\numberwithin{equation}{section}
\usepackage{mathtools}
\begin{document}

\title[Modelling Illiquid Stocks]{Modelling Illiquid Stocks Using Quantum Stochastic Calculus}

\author{Will Hicks}
\address{Will Hicks: Memorial University of Newfoundland, St. John's, NL A1C 5S7, Canada}
\email{williamh@mun.ca}



\begin{abstract}
Quantum Stochastic Calculus can be used as a means by which randomness can be introduced to observables acting on a Hilbert space. In this article we show how the mechanisms of Quantum Stochastic Calculus can be used to extend the classical Black-Scholes framework by incorporating a breakdown in the liquidity of a traded asset. This is captured via the widening of the bid offer spread, and the impact on the nature of the resulting probability distribution is modelled in this work.
\end{abstract}

\maketitle


\section{Introduction}
In this paper we show how to apply the quantum stochastic calculus developed by Hudson \& Parthasarathy in \cite{HP}, to the modelling of illiquid stocks \& the bid offer spread.

Quantum stochastic calculus was first applied to the problem of derivative pricing, and Mathematical Finance, by Accardi \& Boukas in \cite{AB}, where the authors derive a general form for a Quantum Black-Scholes equation. In \cite{Hicks}-\cite{Hicks3}, properties of different example models that can be derived using the Quantum Black-Scholes approach are also investigated.

In the Accardi-Boukas approach, a Hamiltonian controls the drift in the system, and the examples investigated in \cite{AB}, as well as in \cite{Hicks} - \cite{Hicks3}, set the drift based on the classical notion of the self financing hedged portfolio, and the risk free rate. Solutions to the Quantum Black-Scholes models are found by solving a partial differential equation, so that the underlying quantum state, and non-commutative nature of the underlying system, do not impact the result.

The first objective of this research, presented in section \ref{impact of quantum state}, is to investigate the question of what impact the quantum state can have on the modelling outcome. In \cite{Hicks4}, the impact of the system Hamiltonian in a closed quantum system, with no external sources of randomness is investigated. In this article, we consider situations where the nature of the quantum state impacts the stochastic evolution of the traded financial assets under consideration.

As part of this objective, we seek to define different notions of arbitrage in a quantum system, and show how these can be used to derive a derivative pricing equation. The question of the nature of quantum arbitrage has also been addressed in \cite{BR}, where the authors develop a Theorem of asset pricing based on equivalence classes of density operators. In this article we use a similar approach to defining arbitrage, derived by extending the classical definition given in \cite{Bjork}.

The second objective of this research, presented in section \ref{QC_chapter_BO}, is to apply the methods of the quantum stochastic calculus developed in \cite{HP}, to building a quantum stochastic model for the traded asset price that incorporates a bid-offer spread. The approach seeks to illustrate how the Hilbert space framework can be used to develop a model whereby the level of market liquidity, and width of the bid-offer spread, impacts the nature of the dynamics, even in the event that one is able to hedge at the market mid-price.
\section{Outline of the Quantum Approach to Option Pricing:}
\subsection{Quantum Stochastic Processes:}
We model traded financial securities, as operators on a system Hilbert space: $\mathcal{H}$, and assume that the expected price for a particular tradeable asset, represented by the operator: $X$, in a market which we represent by: $|\psi\rangle\in\mathcal{H}$ is given by:
\begin{align*}
E^{\psi}[X]&=\langle\psi|X|\psi\rangle.
\end{align*}
Following the approach outlined in \cite{AB} (see also \cite{HP}), we take the tensor product of $\mathcal{H}$ with the symmetric Fock space: $\mathcal{H}\otimes\Gamma(L^2(\mathbb{R}^+;\mathbb{C}))$, and use a unitary time evolution operator to build the price operator at $t=T$.

If the price operator at $t=0$ is written: $X\otimes\mathbb{I}$, then the operator at $t=T$ is given by: $j_T(X)=U_T^*(X\otimes\mathbb{I})U_T$. $U_t$ is defined by the process (see \cite{HP} proposition 7.1):
\begin{align}\label{U_QSP}
dU_t=-\bigg(\Big(iH+\frac{L^*L}{2}\Big)\otimes dt+L^*S\otimes dA_t-L\otimes dA^{\dagger}_t+(\mathbb{I}-S)\otimes d\Lambda_t\bigg)U_t
\end{align}
Whereby $H,L$, and $S$ act on $\mathcal{H}$, and $dA_t,dA^{\dagger}_t$, and $d\Lambda_t$ act on the Fock space. By writing out (see \cite{HP} Theorem 4.5):\begin{align*}
dj_t(X) &=d(U_t^*(X\otimes\mathbb{I})U_t)\\
&=dU_t^*(X\otimes\mathbb{I})U_t+U_t^*(X\otimes\mathbb{I})dU_t+dU_t^*(X\otimes\mathbb{I})dU_t
\end{align*}
and using It{\^ o} multiplication: Table \ref{ito_table} (see \cite{HP}), we can define a stochastic process for $dj_t(X)$, and $dj_t(X_t)^k, k\geq 2$:
\begin{table}
\centering
\begin{tabular}{p{1cm}|p{1cm}|p{1cm}|p{1cm}|p{1cm}}
-&$dA^{\dagger}_t$&$d\Lambda_t$&$dA_t$&$dt$\\
\hline
$dA^{\dagger}_t$&0&0&0&0\\
$d\Lambda_t$&$dA^{\dagger}_t$&$d\Lambda_t$&0&0\\
$dA_t$&$dt$&$dA_t$&0&0\\
$dt$&0&0&0&0\\
\end{tabular}
\caption{Ito multiplication operators for the basic operators of quantum stochastic calculus.}\label{ito_table}
\end{table}
\begin{align}\label{dX}
dj_t(X)&=j_t(\alpha^{\dagger})dA^{\dagger}_t+j_t(\alpha)dA_t+j_t(\lambda)d\Lambda_t+j_t(\theta)dt\\
k\geq 2:dj_t(X)^k &=j_t(\lambda^{k-1}\alpha^{\dagger}) dA^{\dagger}_t+j_t(\alpha\lambda^{k-1}) dA_t+j_t(\lambda^k) d\Lambda_t+j_t(\alpha\lambda^{k-2}\alpha^{\dagger}) dt\nonumber\\
\theta &=i[H,X]-\frac{1}{2}\Big(L^*LX+XL^*L-2L^*XL\Big)\nonumber\\
\alpha &=[L^*,X]S\nonumber\\
\alpha^{\dagger} &=S^*[X,L]\nonumber\\
\lambda &=S^*XS-X\nonumber
\end{align}
\subsection{Defining the Derivative Price Process}
We define the derivative price process as a self-adjoint operator valued function:
\begin{align*}
V:\mathcal{L}(\mathcal{H}\otimes\Gamma)\times\mathbb{R}^+\rightarrow \mathcal{L}(\mathcal{H}\otimes\Gamma)
\end{align*}
Where $\mathcal{H}$ denotes the system Hilbert space, and $\Gamma$ our choice for the symmetric Fock space. $\mathcal{L}(\mathcal{H})$ represents the space of linear operators on the Hilbert space $\mathcal{H}$.

At each time $t$, the operator $V(j_t(X),t)$ acts on the market state, returning real eigenvalues that represent possible values for the derivative price.

Classically, one proceeds on the basis that any derivative payout can be replicated using a self-financing trading strategy. After the initial investment, no further money needs be invested to replicate the payout. One simply buys \& sells the risky underlying at zero cost at the prevailing market price, and ends up with the same outcome as if one had purchased the derivative. This provides the financial rationale for why the discounted price of the derivative should be a Martingale. Essentially, since one can re-create the payout at zero cost, the expected return on the initial investment should be zero after discounting.

In the quantum case, it is not clear that one can replicate derivative payouts in the same way. In a given market state, both the quantum version of the traded underlying, and the derivative, have uncertain prices. Therefore, even if there exists a formula that outputs the required position (to replicate) based on the current price, since one doesn't know the current price, this is not sufficient.

However, for the time being, we work on the basis that the Martingale price process represents a fair price, before offering a partial justification for this assumption in proposition \ref{arb_just}.
\begin{definition}\label{Mart_Price_Process}
A Martingale Price Process, under the normalised vector $|\psi\rangle\in\mathcal{H}$, is a self-adjoint operator valued map:
\begin{align*}
V:\mathcal{L}(\mathcal{H}\otimes\Gamma)\times\mathbb{R}^+\rightarrow\mathcal{L}(\mathcal{H}\otimes\Gamma)
\end{align*}
Such that:
\begin{align*}
E^{(\psi,\varepsilon)}[V(j_t(X),t)] &=\langle\psi\otimes\varepsilon|V(j_t(X),t)|\psi\otimes\varepsilon\rangle\\
&=E^{(\psi,\varepsilon)}[V(X_0,0)]\\
&=V_0\\
\psi&\in\mathcal{H}\\
\varepsilon&\in\Gamma
\end{align*}
\end{definition}
\begin{remark}
Note that in some cases we take expectations over the initial space: $\mathcal{H}$ say. In these cases, we write:
\begin{align*}
E^{\psi}[\dots]&=\langle\psi|\dots|\psi\rangle
\end{align*}
After a stochastic process has been introduced, we take expectations over the tensor product with the symmetric Fock space: $\mathcal{H}\otimes\Gamma$. In this case we write:
\begin{align*}
E^{(\psi,\varepsilon)}[\dots]&=\langle\psi\otimes\varepsilon|\dots|\psi\otimes\varepsilon\rangle
\end{align*}
\end{remark}
For a fixed time, we can define $V(j_t(X),t)$ using the spectral theorem for self-adjoint operators: \cite{Hall} Theorem 10.4, and the associated functional calculus: \cite{Hall} Definition 10.5. In other words, we write:
\begin{align*}
V_t(j_t(X),t) &=f(j_t(X))\\
\end{align*}
By the Spectral Theorem, there exists a unique projection valued measure: $\mu^X$, such that:
\begin{align}\label{SpectralTheory}
j_t(X)&=\int_{\mathbb{R}}\lambda d\mu^X(\lambda)\nonumber\\
f(j_t(X))&=\int_{\mathbb{R}}f(\lambda) d\mu^X(\lambda)
\end{align}
Where we have assumed the spectrum for the real valued, and unbounded, operator $j_t(X)$ is $\mathbb{R}$. Importantly, we can define a derivative: $\frac{\partial V}{\partial j_t(X)}$ using \ref{SpectralTheory}:
\begin{align*}
\frac{\partial V}{\partial j_t(X)}&=\int_{\mathbb{R}}f'(\lambda)d\mu^X(\lambda)
\end{align*}
Furthermore, assuming sufficient smoothness in the function $V$, we define:
\begin{align*}
\frac{\partial V}{\partial t}&=\lim_{dt\rightarrow 0}\frac{V(j_t(X),t+dt)-V(j_t(X),t)}{dt}
\end{align*}
\subsection{Arbitrage in the Quantum Framework:}\label{arb_section}
Classically, we can define an arbitrage as a derivative price: $V(S,t)$ and a probability measure $P$, such that (see \cite{Bjork} definition 7.5):
\begin{align*}
V(S,0)&=V_0\\
P(V(S,T)>V_0)&>0\\
P(V(S,T)\geq V_0)&=1
\end{align*}
The financial intuition behind this being that at $t=0$ we can execute a trade for free that has the possibility of generating financial gain without the possibility of generating losses.

In the quantum case, under definition \ref{Mart_Price_Process}, $V(j_t(X),t)$ is a self-adjoint operator acting on: $\mathcal{H}\otimes\Gamma$. In order to define what we mean by arbitrage in a quantum model, we first define the following operators, where $f(\lambda)$ is defined by the Spectral Theorem, equation \ref{SpectralTheory}:
\begin{align*}
V_{>0}&=\int_{\mathbb{R}}1_{f(\lambda)>0}(\lambda)d\mu^X(\lambda)\\
V_{\geq 0}&=\int_{\mathbb{R}}1_{f(\lambda)\geq 0}(\lambda)d\mu^X(\lambda)
\end{align*}
We can now define a quantum arbitrage:
\begin{definition}\label{quantum_arb}
Let $V(j_t(X),t)$ be defined by:
\begin{align*}
j_t(X)&=\int_{\mathbb{R}}\lambda d\mu^X(\lambda)\nonumber\\
V(j_t(X),t)&=\int_{\mathbb{R}}f(\lambda) d\mu^X(\lambda)
\end{align*}
For some function $f(\lambda)$. Furthermore, let $|\psi\rangle$ be a normalised vector in the Hilbert space: $\mathcal{H}$. Then $V(j_t(X),t)$ is a quantum arbitrage under $|\psi\rangle$ if:
\begin{align*}
E^{(\psi,\varepsilon)}[V_{>0}]&>0\\
E^{(\psi,\varepsilon)}[V_{\geq 0}]&=1
\end{align*}
\end{definition}
We can now give a definition for non-arbitrage price processes in the quantum framework:
\begin{definition}\label{weak_non_arb}
A weak non-arbitrage price process is a self-adjoint operator valued map:
\begin{align*}
V:\mathcal{L}(\mathcal{H}\otimes\Gamma)\times\mathbb{R}^+\rightarrow\mathcal{L}(\mathcal{H}\otimes\Gamma)
\end{align*}
together with a quantum state represented by the vector, $|\psi\rangle\in\mathcal{H}$, such that there  is no value of $t$, whereby $V(j_t(X),t)$ is a quantum arbitrage under $|\psi\rangle$.
\end{definition}
\begin{definition}\label{strong_non_arb}
A strong non-arbitrage price process is a self-adjoint operator valued map:
\begin{align*}
V:\mathcal{L}(\mathcal{H}\otimes\Gamma)\times\mathbb{R}^+\rightarrow\mathcal{L}(\mathcal{H}\otimes\Gamma)
\end{align*}
Such that there is no normalised vector $|\psi\rangle$, such that $V(j_t(X),t)$ is a weak non-arbitrage price process under $|\psi\rangle$.
\end{definition}

\subsection{Deriving a Partial Differential Equation for V}
\begin{proposition}\label{main_QBS}
Let $V(j_t(X),t)$ represent a Martingale Price Process for a derivative payout. Let $U_t$ be the general time evolution operator defined by \ref{U_QSP} with $X$ defined by equation \ref{dX}. Then we have:
\begin{align}\label{expectation_pde}
E^{(\psi,\varepsilon)}\bigg[\frac{\partial V}{\partial t}+\frac{\partial V}{\partial x}j_t(\theta)+\frac{\partial^2 V}{\partial x^2}\frac{j_t(\alpha\alpha^{\dagger})}{2}\bigg] &=0
\end{align}
Where $\alpha$, and $\theta$, are as in equation \ref{dX}.
\end{proposition}
\begin{proof}
We write out $dV$ as a power series expansion. We denote the arguments of the function $V$ as $x$ and $t$:
\begin{align}\label{dV_powerseries}
dV &= V(j_t(X)+dj_t(X),t+dt)-V(j_t(X),t)\nonumber\\
 &= \sum_{n,k}\frac{1}{n!k!}\frac{\partial^{(n+k)}V}{\partial x^n\partial t^k}(dj_t(X)^n)(dt^k)
\end{align}
Expanding $dV$ using equation \ref{dX}, we have:
\begin{align}\label{dV_exp}
dV &=\bigg(\frac{\partial V}{\partial t}+\frac{\partial V}{\partial x}j_t(\theta)+\frac{\partial^2 V}{\partial x^2}\frac{j_t(\alpha\alpha^{\dagger})}{2}\bigg)dt\nonumber\\
&+\bigg(\frac{\partial V}{\partial x}j_t(\alpha^{\dagger})\bigg)dA^{\dagger}_t\nonumber\\
&+\bigg(\frac{\partial V}{\partial x}j_t(\alpha)\bigg)dA_t
\end{align}
\newline
After taking expectations and equating to zero, we are left with the $dt$ terms. Equating these to zero leaves:
\begin{align*}
E^{(\psi,\varepsilon)}\bigg[\frac{\partial V}{\partial t}+\frac{\partial V}{\partial x}j_t(\theta)+\frac{\partial^2 V}{\partial x^2}\frac{j_t(\alpha\alpha^{\dagger})}{2}\bigg] &=0
\end{align*}
\end{proof}
\begin{proposition}\label{classical_QBS}
Let the system space $\mathcal{H}=L^2(\mathbb{R})$, and $X|\psi(x)\rangle =|x\psi(x)\rangle$, for $|\psi(x)\rangle\in L^2(\mathbb{R})$. We set $L=-i\sigma\frac{\partial}{\partial x}$, and $S=\mathbb{I}$. Finally, we set: $H=i(\sigma^2/2)\frac{\partial}{\partial x}$. Then equation \ref{expectation_pde} becomes:
\begin{align}\label{Classical_pde}
E^{(\psi,\varepsilon)}\bigg[\frac{\partial V}{\partial t}-\frac{\sigma^2}{2}\frac{\partial V}{\partial x}+\frac{\sigma^2}{2}\frac{\partial^2 V}{\partial x^2}\bigg] &=0
\end{align}
\end{proposition}
\begin{proof}
First note that we have:
\begin{align*}
\alpha&=[L^*,X]\\
&=-i\sigma
\end{align*}
Further note that:
\begin{align*}
XL^*L+L^*LX-2L^*XL&=0
\end{align*}
Therefore:
\begin{align*}
\theta&=i[H,X]\\
&=-\frac{\sigma^2}{2}
\end{align*}
Feeding this into \ref{expectation_pde} gives the required result.
\end{proof}
\begin{proposition}\label{arb_just}
The solution to equation \ref{expectation_pde}, is a weak non-arbitrage price process. Where there is a solution to the underlying PDE: \ref{Classical_pde}, this is a strong non-arbitrage price process.
\end{proposition}
\begin{proof}
Assume, $V(j_t(X),t)$ is a quantum arbitrage under $|\psi\rangle$, and assume that:
\begin{align*}
E^{(\psi,\varepsilon)}[V(j_0(X),0)]&=0
\end{align*}
From equation \ref{dV_exp}, we have:
\begin{align*}
dV &=\bigg(\frac{\partial V}{\partial t}+\frac{\partial V}{\partial x}j_t(\theta)+\frac{\partial^2 V}{\partial x^2}\frac{j_t(\alpha\alpha^{\dagger})}{2}\bigg)dt\\
&+\bigg(\frac{\partial V}{\partial x}j_t(\alpha^{\dagger})\bigg)dA^{\dagger}_t+\bigg(\frac{\partial V}{\partial x}j_t(\alpha)\bigg)dA_t
\end{align*}
Now take expectations over the Hilbert space: $\mathcal{H}$ to get:
\begin{align}\label{dV_noise_mart}
E^{(\psi,\varepsilon)}[dV] &=E^{\psi}\bigg[\bigg(\frac{\partial V}{\partial t}+\frac{\partial V}{\partial x}j_t(\theta)+\frac{\partial^2 V}{\partial x^2}\frac{j_t(\alpha\alpha^{\dagger})}{2}\bigg)\bigg]dt\nonumber\\
&+E^{\psi}\bigg[\bigg(\frac{\partial V}{\partial x}j_t(\alpha^{\dagger})\bigg)\bigg]dA^{\dagger}_t+E^{\psi}\bigg[\frac{\partial V}{\partial x}\bigg(j_t(\alpha)\bigg)\bigg]dA_t\nonumber\\
&=E^{\psi}\bigg[\bigg(\frac{\partial V}{\partial x}j_t(\alpha^{\dagger})\bigg)\bigg]dA^{\dagger}_t+E^{\psi}\bigg[\bigg(\frac{\partial V}{\partial x}j_t(\alpha)\bigg)\bigg]dA_t
\end{align}
In fact, since $V(j_t(X),t)$ is a Martingale price process, we have $E^{\psi}[dV]=0$. Thus, after taking expectations over the Hilbert space $\mathcal{H}$, $V(j_t(X),t)$ is the zero process, with no stochastic noise. Therefore we have: $E^{\psi,\varepsilon}[V_{>0}]=0$, which contradicts the requirements for a quantum arbitrage.

Finally, if the PDE \ref{Classical_pde} is met, then this conditions holds regardless of the quantum state, and $V(j_t(X),t)$ must therefore be a strong non-arbitrage process.
\end{proof}
\section{Modelling Bid-Offer Spread:}\label{impact of quantum state}
\subsection{Hilbert Space Representation of the Market:}
We consider the market as being made up of a number of buyers who would like to buy at the lower {\em bid} price, and sellers who would like to sell at the higher {\em offer} price.

Furthermore, we consider the case where there are 2 state variables. One coordinate: $x$, that represents the mid-price for the traded asset, and a second coordinate: $\epsilon$ that represents the width of the bid-offer spread.

We therefore assume that the state of the market for potential buyers (and sellers) is determined by wave functions in the Hilbert space of complex valued square integrable functions on $\mathbb{R}^2$:
\begin{align}\label{psi_o,b}
\psi_o(x,\epsilon)\in L^2(\mathbb{R}^2,\mathbb{C})\\
\psi_b(x,\epsilon)\in L^2(\mathbb{R}^2,\mathbb{C})\nonumber
\end{align}
The question now arises as to how to combine the two into a single system space for the market.
\subsection{Tensor Product Space:}
Instead of the coordinates $x$ and $\epsilon$, we could consider 2 variables as the bid price and the offer price. In place of equation \ref{psi_o,b}, we could apply: $\psi_o(x_o)\otimes\psi_b(x_b)\in L^2(\mathbb{R}^2)$. Bounded operators on: $L^2(\mathbb{R})\otimes L^2(\mathbb{R})\cong L^2(\mathbb{R}^2)$ can be defined as the algebraic tensor product: $B(L^2(\mathbb{R}))\otimes B(L^2(\mathbb{R}))$.
\begin{align*}
(A\otimes B)(\psi_o\otimes\psi_b) &=A\psi_o\otimes B\psi_b
\end{align*}
Thus whilst interaction between the two spaces is possible, we have one set of operators that acts on the first space, and a second that acts on the second space.

In fact, unlike wave functions of coordinates in Euclidean space, the bid and offer wave functions act on the same axis, and we would like to be able to define operators that act on both:
\begin{align*}
A\psi_o\text{, }A\psi_b
\end{align*}
Furthermore, whilst the buyers that make up a market may wish to transact at a low price and the sellers that make up a market may wish to transact at a higher price, the 2 may occasionally meet in the middle. In other words, we would like inner products like: $\langle\psi_o|\psi_b\rangle$ to have meaning.

Finally, the question arises as to how we wish the state to be normalised. If we have: $\psi_o\otimes\psi_b$, then the state is normalised via:
\begin{align*}
||\psi_o\otimes\psi_b||^2 &=\langle\psi_o\otimes\psi_b|\psi_o\otimes\psi_b\rangle\\
&= \langle\psi_o|\psi_o\rangle\langle\psi_b|\psi_b\rangle\\
&=||\psi_o||^2||\psi_b||^2\\
&=1
\end{align*}
The question now arises regarding the situation whereby there is a tradeable asset, such as shares in a more or less defunct company, where there are no buyers. In this extreme case, we would still like to be able to define a normalisable market state.

For this reason, we look to representing the Hilbert state using the direct sum, rather than the tensor product.
\subsection{Hilbert Space Direct Sum:}\label{Hilbert_space_direct_sum}
We define:
\begin{align}\label{Market_def}
\psi &=\psi_o\oplus\psi_b\\
\psi_o(x,\epsilon)\text{, }\psi_b(x,\epsilon) &\in L^2(\mathbb{R}^2)\nonumber
\end{align}
For $\phi=\phi_1\oplus\phi_2$ and $\psi=\psi_1\oplus\psi_2$, we have:
\begin{align*}
\langle\phi|\psi\rangle &=\langle\phi_1|\psi_1\rangle+\langle\phi_2|\psi_2\rangle
\end{align*}
So it follows that the normalisation condition becomes:
\begin{align}\label{norm_cond}
||\psi_0\oplus\psi_b||^2&=||\psi_o||^2+||\psi_b||^2\nonumber\\
&=1
\end{align}
For example, we may have an even balance of buyers \& sellers, in which case:
\begin{align*}
||\psi_o||^2=||\psi_b||^2=1/2
\end{align*}
In general, as long as the normalization condition, given by equation \ref{norm_cond}, is met then we can have:
\begin{itemize}
\item More buyers than sellers: $||\psi_b||^2>||\psi_o||^2$.
\item More sellers than buyers: $||\psi_o||^2>||\psi_b||^2$.
\end{itemize}
\begin{remark}
Going forward, we make use of matrix notation, so that for $\psi\in \mathcal{S}(\mathbb{R})\oplus\mathcal{S}(\mathbb{R})$ we write:
\begin{align*}
|\psi\rangle &=\begin{pmatrix}\psi_0\\ \psi_b\end{pmatrix}\\
A\psi&=\begin{pmatrix} A_{11}& A_{12}\\A_{21}&A_{22}\end{pmatrix}\begin{pmatrix}\psi_0\\ \psi_b\end{pmatrix}
\end{align*}
Note that, we also apply the following abuse of notation, by writing:
\begin{align*}
\langle\psi|&=\begin{pmatrix}\psi_o&\psi_b\end{pmatrix}
\end{align*}
So that we write:
\begin{align*}
E^{\psi}[A]&=\langle\psi|A|\psi\rangle\\
&=\begin{pmatrix}\psi_o&\psi_b\end{pmatrix}\begin{pmatrix} A_{11}& A_{12}\\A_{21}&A_{22}\end{pmatrix}\begin{pmatrix}\psi_0\\ \psi_b\end{pmatrix}\\
&=\langle\psi_o|A_{11}|\psi_o\rangle+\langle\psi_o|A_{12}|\psi_b\rangle+\langle\psi_b|(A_{21}|\psi_o\rangle+\langle\psi_b|A_{22}|\psi_b\rangle
\end{align*}
\end{remark}
\subsection{Price Operator:}
We now define the price operator $X$:
\begin{align}\label{X_com}
X=\begin{pmatrix} x+\epsilon/2&0\\0&x-\epsilon/2\end{pmatrix}
\end{align}
Consider the example whereby we have the following evenly balance market of buyers \& sellers:
\begin{align}\label{bivar}
|\psi_o(x,\epsilon)|^2=(1/2)N(x,\epsilon,\mu,\Sigma)\\
|\psi_b(x,\epsilon)|^2=(1/2)N(x,\epsilon,\mu,\Sigma)\nonumber\\
\mu=\begin{pmatrix}x_{mid}\\ \epsilon_0\end{pmatrix}\nonumber\\
\Sigma=\begin{pmatrix}\sigma_x&0\\0&\sigma_{\epsilon}\end{pmatrix}\nonumber
\end{align}
Before introducing a stochastic process, the weighted average price is given by:
\begin{align*}
E^{\psi}[X] &=\int_{\mathbb{R}^2}(x+\epsilon/2)|\psi_o(x,\epsilon)|^2dxd\epsilon+\int_{\mathbb{R}^2}(x-\epsilon/2)|\psi_b(x,\epsilon)|^2dxd\epsilon\\
&=x_{mid}
\end{align*}
We can consider a situation whereby a sudden piece of adverse news causes a rush to sell (even if it means selling at the lower bid price) using a rotation:
\begin{align}\label{bear_theta}
R(\theta) &=\begin{pmatrix} \cos(\theta)&-\sin(\theta)\\\sin(\theta)&\cos(\theta)\end{pmatrix}\\
\begin{pmatrix}\psi_o\\\psi_b\end{pmatrix} &\rightarrow\begin{pmatrix} \cos(\theta)&-\sin(\theta)\\\sin(\theta)&\cos(\theta)\end{pmatrix}\begin{pmatrix}\psi_o\\\psi_b\end{pmatrix}\nonumber
\end{align}
Now if we calculate the weighted average price after the market rotation from buyers to sellers, we find:
\begin{align*}
E^{\psi}[R(\theta)^*XR(\theta)]&=\begin{pmatrix}\psi_o&\psi_b\end{pmatrix}\begin{pmatrix} x+\cos(2\theta)\epsilon/2 & -\sin(2\theta)\epsilon/2\\-\sin(2\theta)\epsilon/2 & x-\cos(2\theta)\epsilon/2\end{pmatrix}\begin{pmatrix} \psi_o\\ \psi_b\end{pmatrix}\\
&=x_{mid}-\sin(2\theta)\epsilon_0/2
\end{align*}
So that a market rotation from buyers to sellers (represented by equation \ref{bear_theta}) causes a fall in the expected trade price.
\subsection{Introducing a Gaussian Stochastic Process:}\label{first}
The functional form for the stochastic process, derived in \cite{AB}, is given by equation \ref{dX}. The first question is what to use for the linear operator: $L$. Some possibilities:
\begin{itemize}
\item A single factor stochastic process, with $\partial_x$ operators only. In this case the model essentially reduces to the simple quantum model above. Whilst $\epsilon$ is a Euclidean coordinate for the system, it is a static one that does not evolve.
\item A 2 factor stochastic process, with separate operators; $L_x$ and $L_{\epsilon}$.
\item A single factor stochastic process, which incorporates volatility in both the $x$ and $\epsilon$ variables. In this case, $L$ must incorporate both $\partial_x$ and $\partial_{\epsilon}$ operators.
\end{itemize}
Whilst there is merit in the second approach (2 factor approach), for the time being we stick with a 1 factor stochastic process, and set:
\begin{align}\label{ext_L}
L &=\begin{pmatrix} -i\sigma_x\partial_x-i\sigma_{\epsilon}\partial_{\epsilon}&0\\0&-i\sigma_x\partial_x-i\sigma_{\epsilon}\partial_{\epsilon}\end{pmatrix}
\end{align}
Furthermore, to keep the approach as general as possible, we set:
\begin{align}\label{ext_S}
R(\theta) &=\begin{pmatrix} \cos(\theta)&-\sin(\theta)\\\sin(\theta)&\cos(\theta)\end{pmatrix}
\end{align}
and incorporate $R(\theta)$ directly into the $X$ operator:
\begin{align}\label{ext_X}
X &=\begin{pmatrix} x+\cos(2\theta)\epsilon/2&-\sin(2\theta)\epsilon/2\\-\sin(2\theta)\epsilon/2&x-\cos(2\theta)\epsilon/2\end{pmatrix}
\end{align}
So that now, the price operator incorporates a market rotation between buyers \& sellers. In this way a small trade, unlikely to significantly effect the best available market price, could be represented by $\theta=0$. A large sell order initiated by a bearish market speculator, that would likely cause an adverse impact on the market price, would be represented by $\theta>0$.
\begin{proposition}\label{ext_var1}
Let the stochastic process: $dj_t(X)$ be defined by equation \ref{dX}, with $L$ given by: \ref{ext_L}, and $X$ given by: \ref{ext_X}. Then, the variance is given by:
\begin{align}
E^{(\psi,\varepsilon)}\big[j_t(X)^2\big]&=\big(\sigma_x^2t+(\sigma_{\epsilon}^2/4)t+\cos(2\theta)\sigma_x\sigma_{\epsilon}t\big)||\psi_o||^2\nonumber\\
&+\big(\sigma_x^2t+(\sigma_{\epsilon}^2/4)t-\cos(2\theta)\sigma_x\sigma_{\epsilon}t\big)||\psi_b||^2\nonumber\\
&+\sin(2\theta)\big(\langle\psi_b|\psi_o\rangle+\langle\psi_o|\psi_b\rangle\big)\sigma_x\sigma_{\epsilon}t
\end{align}
\end{proposition}
\begin{proof}
We must calculate $\alpha$, and $\theta$. Applying \ref{ext_L} and \ref{ext_X}, we get:
\begin{align*}
[L^*,X]&=\begin{pmatrix} i\sigma_x+i\cos(2\theta)\sigma_{\epsilon}&-i\sin(2\theta)\sigma_{\epsilon}\\-i\sin(2\theta)\sigma_{\epsilon}&i\sigma_x-i\cos(2\theta)\sigma_{\epsilon}\end{pmatrix}\\
XL^*L+L^*LX-2L^*XL &=0
\end{align*}
Inserting this into \ref{dX}, with $S=\mathbb{I}$ we get:
\begin{align}\label{dX_ext}
dj_t(X)&=\begin{pmatrix}i\sigma_x+i\cos(2\theta)\sigma_{\epsilon}/2&-i\sin(2\theta)\sigma_{\epsilon}/2\\-i\sin(2\theta)\sigma_{\epsilon}/2&i\sigma_x-i\cos(2\theta)\sigma_{\epsilon}/2\end{pmatrix}dA^{\dagger}_t\\
&-\begin{pmatrix}i\sigma_x+i\cos(2\theta)\sigma_{\epsilon}/2&-i\sin(2\theta)\sigma_{\epsilon}/2\\-i\sin(2\theta)\sigma_{\epsilon}/2&i\sigma_x-i\cos(2\theta)\sigma_{\epsilon}/2\end{pmatrix}dA_t\nonumber\\
dj_t(X)^2&=\begin{pmatrix}\sigma_x^2+\sigma_{\epsilon}^2/4+\cos(2\theta)\sigma_x\sigma_{\epsilon}&\sin(2\theta)\sigma_x\sigma_{\epsilon}\\ \sin(2\theta)\sigma_x\sigma_{\epsilon}&\sigma_x^2+\sigma_{\epsilon}^2/4-\cos(2\theta)\sigma_x\sigma_{\epsilon}\end{pmatrix}dt\nonumber
\end{align}
The stochastic integrals for $j_t(X)$, and $j_t(X)^2$ can be defined using definition 4.1 in \cite{HP}:
\begin{align*}
j_t(X)&=\int_0^t\begin{pmatrix}i\sigma_x+i\cos(2\theta)\sigma_{\epsilon}/2&-i\sin(2\theta)\sigma_{\epsilon}/2\\-i\sin(2\theta)\sigma_{\epsilon}/2&i\sigma_x+i\cos(2\theta)\sigma_{\epsilon}/2\end{pmatrix}dA^{\dagger}_s\\
&-\int_0^t\begin{pmatrix}i\sigma_x+i\cos(2\theta)\sigma_{\epsilon}/2&-i\sin(2\theta)\sigma_{\epsilon}/2\\-i\sin(2\theta)\sigma_{\epsilon}/2&i\sigma_x+i\cos(2\theta)\sigma_{\epsilon}/2\end{pmatrix}dA_s\\
j_t(X)^2&=\int_0^t\begin{pmatrix}\sigma_x^2+\sigma_{\epsilon}^2/4+\cos(2\theta)\sigma_x\sigma_{\epsilon}&\sin(2\theta)\sigma_x\sigma_{\epsilon}\\\sin(2\theta)\sigma_x\sigma_{\epsilon}&\sigma_x^2+\sigma_{\epsilon}^2/4-\cos(2\theta)\sigma_x\sigma_{\epsilon}\end{pmatrix}ds
\end{align*}
The operators $dA_t$ and $dA^{\dagger}_t$ act on exponential vectors in the symmetric Fock space: $\Gamma(L^2(\mathbb{R}^+;\mathbb{C}))$. For example:
\begin{align*}
e(u) &=\bigoplus_{n=0}^{\infty}(n!)^{-1/2}u^{\otimes n}\\
u &\in L^2(\mathbb{R}^+;\mathbb{C})
\end{align*}
We write the Hilbert space vector for the market as:
\begin{align*}
|\psi\otimes ce(u)\rangle
\end{align*}
Where $c$ is a normalising constant:
\begin{align*}
c^2&=\frac{1}{\langle e(u)|e(u)\rangle}
\end{align*}
Therefore we can write:
\begin{align*}
E^{(\psi,e(u))}[j_t(X)] &=c^2\langle\big(\psi\otimes e(u)\big)|j_t(X)\big(\psi\otimes e(u)\big)\rangle\\
E^{(\psi,e(u))}[j_t(X)^2] &=c^2\langle\big(\psi\otimes e(u)\big)|j_t(X)^2\big(\psi\otimes e(u)\big)\rangle\\
\psi &=\begin{pmatrix}\psi_o\\ \psi_b\end{pmatrix}
\end{align*}
In fact, we have (see for example \cite{HP}):
\begin{align*}
\Big\langle ce(u)\Big|\Big(\int_0^tds\Big)ce(u)\Big\rangle &=c^2\int_0^t1ds\langle e(u)|e(u)\rangle\\
&=t
\end{align*}
Furthermore, we have:
\begin{align*}
E^{(\psi,e(u))}[j_t(X)]&=X_0
\end{align*}
So, conditional on $X_0=0$, we have:
\begin{align*}
Var(X)&=E^{(\psi,\varepsilon)}[j_t(X)^2]\nonumber\\
&=\begin{pmatrix}\psi_0&\psi_b\end{pmatrix}\begin{pmatrix}\sigma_x^2+\frac{\sigma_{\epsilon}^2}{4}+\cos(2\theta)\sigma_x\sigma_{\epsilon}&\sin(2\theta)\sigma_x\sigma_{\epsilon}\\\sin(2\theta)\sigma_x\sigma_{\epsilon}&\sigma_x^2+\frac{\sigma_{\epsilon}^2}{4}-\cos(2\theta)\sigma_x\sigma_{\epsilon}\end{pmatrix}\begin{pmatrix}\psi_o\\\psi_b\end{pmatrix}t\\
&=\big(\sigma_x^2+(\sigma_{\epsilon}^2/4)\big)t+\cos(2\theta)\sigma_x\sigma_{\epsilon}\big(||\psi_o||^2-||\psi_b||^2\big)t\\
&+\sin(2\theta)\big(\langle\psi_b|\psi_o\rangle+\langle\psi_o|\psi_b\rangle\big)\sigma_x\sigma_{\epsilon}t
\end{align*}
\end{proof}
\begin{remark}
Importantly, due to the particular choices for the $L$ and $X$ operators (equations \ref{ext_L} and \ref{ext_X}), we end in a situation where it is no longer possible to define $V$ such that:
\begin{align*}
\frac{\partial V}{\partial t}+\frac{\partial V}{\partial x}j_t(\theta)+\frac{\partial^2 V}{\partial x^2}\frac{j_t(\alpha\alpha^{\dagger})}{2}
\end{align*}
is the zero operator. Crucially, we need to take account of the quantum state to solve:
\begin{align*}
E^{(\psi,\varepsilon)}\bigg[\frac{\partial V}{\partial t}+\frac{\partial V}{\partial x}j_t(\theta)+\frac{\partial^2 V}{\partial x^2}\frac{j_t(\alpha\alpha^{\dagger})}{2}\bigg]=0
\end{align*}
\end{remark}
\subsubsection{Financial Interpretation of the Extended Quantum Approach:}\label{state_dep_vol}
First consider the situation, whereby $\psi_o,\psi_b$ are given by \ref{bivar}. In other words, there is an even balance of market buyers and sellers. In this case we have:
\begin{align*}
\langle\psi_o|\psi_b\rangle &=\langle\psi_b|\psi_o\rangle\\
&=||\psi_o||^2=||\psi_b||^2=1/2
\end{align*}
Therefore, from equation \ref{ext_var1}, we get:
\begin{align*}
Var(X)&=(\sigma_x^2+\sigma_{\epsilon}^2)t+\sin(2\theta)\sigma_x\sigma_{\epsilon}
\end{align*}
Now assume we enter a bear market sentiment, ie sellers being willing to sell at a lower than normal price. We can see from  equation \ref{bear_theta}, that this situation is represented by $\theta>0$, and we have an increase in the volatility.

Similarly, if we enter a bull market sentiment, with buyers being willing to pay a higher price, then $\theta<0$, and we have a reduction in the volatility. In other words, the model has a natural skew, similar to that observed in most implied volatility surfaces.
\section{A Non-Gaussian Model For Bid-Offer Spread:}\label{QC_chapter_BO}
In this section, we show how the incorporation of $R(\theta)$ into the unitary time evolution operator defined by \ref{U_QSP}, leads to a non-Gaussian process. We go on to discuss the financial interpretation of doing so, in the next section. First, we set $S=R(\pi/2)$, in equation \ref{U_QSP}: 
\begin{align*}
\lambda&=S^*XS-X\\
&=\begin{pmatrix} x+\cos(\pi)\epsilon/2 & -\sin(\pi)\epsilon/2\\-\sin(\pi)\epsilon/2 & x-\cos(\pi)\epsilon/2\end{pmatrix}-\begin{pmatrix} x+\epsilon/2 & 0\\0&x-\epsilon/2\end{pmatrix}\\
&=\begin{pmatrix} x-\epsilon/2&0\\0&x+\epsilon/2\end{pmatrix}-\begin{pmatrix} x+\epsilon/2 & 0\\0&x-\epsilon/2\end{pmatrix}\\
&=\begin{pmatrix} -\epsilon&0\\0&\epsilon\end{pmatrix}
\end{align*}
In section \ref{first}, we incorporated derivatives with respect to both the mid-price \& bid-offer spread variables. This enabled the development of a model whereby the variance of the mid-price depended on the market quantum state.

In this section, we simplify the stochastic process, so that only the mid-price evolves stochastically. However, we show how by incorporating the market rotation from buyers (represented by: $\psi_b$) to sellers (represented by $\psi_0$) into equation \ref{U_QSP}, we generate a non-Gaussian model. With this in mind, we change the $L$ operator to:
\begin{align*}
L\begin{pmatrix} \psi_o\\ \psi_b\end{pmatrix} &=\begin{pmatrix} -i\sigma\partial_x&0\\0& -i\sigma\partial_x\end{pmatrix}\begin{pmatrix} \psi_o\\ \psi_b\end{pmatrix}
\end{align*}
Inserting this into equation \ref{dX}, we get the following stochastic process for $j_t(X)$, $j_t(X)^k$:
\begin{align}\label{dX_S}
dj_t(X) &=\begin{pmatrix}0&i\sigma\\-i\sigma&0\end{pmatrix}dA_t+\begin{pmatrix}0&i\sigma\\-i\sigma&0\end{pmatrix}dA^{\dagger}_t+j_t\begin{pmatrix}-\epsilon & 0\\0&\epsilon\end{pmatrix}d\Lambda_t\\
dj_t(X)^k &=\Bigg(\begin{pmatrix}\sigma^2&0\\0&\sigma^2\end{pmatrix}j_t\begin{pmatrix}-\epsilon & 0\\0&\epsilon\end{pmatrix}^{k-2}\Bigg)dt\nonumber\\
&+\Bigg(\begin{pmatrix}0&i\sigma\\-i\sigma&0\end{pmatrix}j_t\begin{pmatrix}-\epsilon & 0\\0&\epsilon\end{pmatrix}^{k-1}\Bigg)dA_t\nonumber\\
&+\Bigg(j_t\begin{pmatrix}-\epsilon & 0\\0&\epsilon\end{pmatrix}^{k-1}\begin{pmatrix}0&i\sigma\\-i\sigma&0\end{pmatrix}\Bigg)dA^{\dagger}_t\nonumber\\
&+\Bigg(j_t\begin{pmatrix}-\epsilon & 0\\0&\epsilon\end{pmatrix}\Bigg)^kd\Lambda_t\nonumber
\end{align}
\subsection{Deriving The Price Process:}
As above, we model the derivative price as an operator valued function:
\begin{align*}
V:\mathcal{L}(\mathcal{H}\otimes\Gamma)\times\mathbb{R}^+\rightarrow \mathcal{L}(\mathcal{H}\otimes\Gamma)
\end{align*}
The goal is to solve for a derivative price $V(j_t(X),t)$ such that:
\begin{align*}
E^{(\psi,\varepsilon)}[V(j_t(X),t)] &=\langle(\psi\otimes\varepsilon)|V(j_t(X),t)|(\psi\otimes\varepsilon)\rangle\\
&=E^{(\psi,\varepsilon)}[V(j_0(X),0)]\nonumber\\
&=V_0
\end{align*}
\begin{proposition}
Let $V(j_t(X),t)$ represent a Martingale Price Process for a derivative payout. Let $U_t$ be the time evolution operator defined by \ref{U_QSP} with $X$ defined by equation \ref{ext_X}, $L$ defined by equation \ref{ext_L}, $S=R(\pi/2)$, and $H=0$. Then we have:
\begin{align}
\frac{\partial V}{\partial t}&+\frac{\sigma^2}{2}\frac{\partial^2 V}{\partial x^2}+\sigma^2\sum_{k\geq 2}\frac{\epsilon^{(2k-2)}}{(2k)!}\frac{\partial^{2k}V}{\partial x^{2k}}\\
&+\sigma^2(||\psi_b||^2-||\psi_o||^2)\sum_{k\geq 2}\frac{(-\epsilon)^{(2k-3)}}{(2k-1)!}\frac{\partial^{(2k-1)}V}{\partial x^{(2k-1)}}=0\nonumber
\end{align}
\end{proposition}
\begin{proof}
Expanding incremental changes in $V$ as a power series we get:
\begin{align*}
dV &\coloneqq V(j_t(X)+dj_t(X),t+dt)-V(j_t(X),t)\\
 &= \sum_{n,k}\frac{1}{n!k!}\frac{\partial^{(n+k)}V}{\partial x^n\partial t^k}(dj_t(X)^n)(dt^k)
\end{align*}
We then insert from \ref{dX_S} to get:
\begin{align}\label{dV_S}
dV &=\bigg(\frac{\partial V}{\partial t}+\sum_{k\geq 2}\frac{1}{k!}\frac{\partial^k V}{\partial x^k}\begin{pmatrix}\sigma^2&0\\0&\sigma^2\end{pmatrix}\begin{pmatrix}-\epsilon & 0\\0&\epsilon\end{pmatrix}^{k-2}\bigg)dt\\
&+\bigg(\sum_{k\geq 1}\frac{\partial^k V}{\partial x^k}\begin{pmatrix}0&i\sigma\\-i\sigma&0\end{pmatrix}\begin{pmatrix}-\epsilon & 0\\0&\epsilon\end{pmatrix}^{k-1}\bigg)dA_t\nonumber\\
&-\bigg(\sum_{k\geq 1}\frac{\partial^k V}{\partial x^k}\begin{pmatrix}-\epsilon & 0\\0&\epsilon\end{pmatrix}^{k-1}\begin{pmatrix}0&i\sigma&0\\-i\sigma&0\end{pmatrix}\bigg)dA^{\dagger}_t\nonumber\\
&+\bigg(\sum_{k\geq 1}\frac{\partial^k V}{\partial x^k}\begin{pmatrix}0&i\sigma\\-i\sigma&0\end{pmatrix}\begin{pmatrix}-\epsilon & 0\\0&\epsilon\end{pmatrix}^{k}\begin{pmatrix}0&i\sigma\\-i\sigma&0\end{pmatrix}\bigg)d\Lambda_t\nonumber
\end{align}
Taking expectations and equating to zero, we find that:
\begin{align*}
E^{(\psi,\varepsilon)}\bigg[\frac{\partial V}{\partial t}+\sum_{k\geq 2}\frac{1}{k!}\frac{\partial^k V}{\partial x^k}\begin{pmatrix}\sigma^2&0\\0&\sigma^2\end{pmatrix}\begin{pmatrix}(-\epsilon)^{k-2}&0\\0&\epsilon^{k-2}\end{pmatrix}\bigg]&=0
\end{align*}
Setting:
\begin{align*}
|\psi\rangle&=\begin{pmatrix}\psi_0\\ \psi_b\end{pmatrix}
\end{align*}
We get:
\begin{align}\label{NG_QBS}
\frac{\partial V}{\partial t}&+\frac{\sigma^2}{2}\frac{\partial^2 V}{\partial x^2}+\sigma^2\sum_{k\geq 2}\frac{\epsilon^{(2k-2)}}{(2k)!}\frac{\partial^{2k}V}{\partial x^{2k}}\\
&+\sigma^2(||\psi_b||^2-||\psi_o||^2)\sum_{k\geq 2}\frac{(-\epsilon)^{(2k-3)}}{(2k-1)!}\frac{\partial^{(2k-1)}V}{\partial x^{(2k-1)}}=0\nonumber
\end{align}
\end{proof}
\subsection{Associated Fokker-Planck Equation:}
In this section, we calculate the Fokker-Planck equation associated to the partial differential equation \ref{NG_QBS}. We follow the basic strategy outlined in \cite{Oksendal}. See also \cite{Hicks} proposition 3.1.

First of all note that, since: $E^{(\psi,\varepsilon)}[dV]=0$, we have that:
\begin{align*}
V_0=E^{(\psi,\varepsilon)}[V(j_T(X),T]
\end{align*}
If we are working with a derivative that pays out at final maturity $T$, then $V_(j_T(X),T)$ is defined by the contractual payout:
\begin{align*}
V_T(j_T(X),T)&=\chi(j_T(X))
\end{align*}
Thus we have:
\begin{align*}
V_0=E^{(\psi,\varepsilon)}[\chi(j_T(X)]
\end{align*}
For this reason, one can associate to equation \ref{NG_QBS} a Fokker-Planck equation for the underlying probability function.
\begin{proposition}\label{NG_FP_prop}
The Fokker-Planck equation associated to the partial differential equation \ref{NG_QBS} is given by:
\begin{align}\label{NG_FP}
\frac{\partial p}{\partial t}=\frac{\sigma^2}{2}\frac{\partial^2 p}{\partial x^2}+\sigma^2\sum_{k\geq 2}\frac{\epsilon^{(2k-2)}}{(2k)!}\frac{\partial^{2k} p}{\partial x^{2k}}+\sigma^2\big(||\psi_o||^2-||\psi_b||^2\big)\sum_{k\geq 2}\frac{(-\epsilon)^{(2k-3)}}{(2k-1)!}\frac{\partial^{(2k-1)} p}{\partial x^{(2k-1)}}
\end{align}
\end{proposition}
\begin{proof}
Applying the Spectral Theorem, we have (for some probability density function $p(y,t)$):
\begin{align}\label{expchi}
V_0&=E^{(\psi,\varepsilon)}[V(j_T(X),T)]\nonumber\\
&=\int_{\mathbb{R}}\chi(y)p(y,T)dy
\end{align}
First expand the function:
\begin{align*}
d\chi &\coloneqq \chi(j_t(X)+dj_t(X))-\chi(j_t(X))
\end{align*}
as a power series in $dj_t(X)$:
\begin{align}\label{chi_powerseries}
d\chi &=\sum_{k}\frac{1}{k!}\frac{\partial^k\chi}{\partial x^k}dj_t(X)^k
\end{align}
After taking expectations over the symmetric Fock space, we find (for some unknown probability density function $p(x,t)$):
\begin{align*}
E^{(\psi,\varepsilon)}[d\chi]&=\int_{\mathbb{R}}E^{\psi}\big[d\chi(y)\big]p(y,t)dy
\end{align*}
Now, applying equation \ref{dX_S} to equation \ref{chi_powerseries}, then integrating from $0$ to $T$, we get:
\begin{align}\label{interim_int}
E^{(\psi,\varepsilon)}[\chi(j_T(X))]&=\chi_0\nonumber\\
&+\int_0^T\int_{\mathbb{R}}E^{\psi}\bigg[\sum_{k\geq 2}\frac{1}{k!}\frac{\partial^k\chi}{\partial x^k}\begin{pmatrix}\sigma^2&0\\0&\sigma^2\end{pmatrix}\begin{pmatrix}-\epsilon & 0\\0&\epsilon\end{pmatrix}^{k-2}\bigg]p(y,t)dydt
\end{align}
So we have:
\begin{align*}
E^{(\psi,\varepsilon)}[\chi(j_T(X))]&=\chi_0\\
&+\int_0^T\int_{\mathbb{R}}\bigg(\sigma^2\sum_{k\geq 2}\frac{\epsilon^{(2k-2)}}{(2k)!}\frac{\partial^{2k}\chi}{\partial x^{2k}}\\
&+\sigma^2\big(||\psi_b||^2-||\psi_0||^2\big)\sum_{k\geq 2}\frac{(-\epsilon)^{(2k-3)}}{(2k-1)!}\frac{\partial^{(2k-1)}\chi}{\partial x^{(2k-1)}}\bigg)p(t,y)dydt
\end{align*}
Inserting from equation \ref{expchi}, we get:
\begin{align*}
\int_{\mathbb{R}}\chi(y)p(y,T)dy&=\int_0^T\int_{\mathbb{R}}\bigg(\sigma^2\sum_{k\geq 2}\frac{\epsilon^{(2k-2)}}{(2k)!}\frac{\partial^{2k}\chi}{\partial x^{2k}}\\
&+\sigma^2\big(||\psi_b||^2-||\psi_0||^2\big)\sum_{k\geq 2}\frac{(-\epsilon)^{(2k-3)}}{(2k-1)!}\frac{\partial^{(2k-1)}\chi}{\partial x^{(2k-1)}}\bigg)p(t,y)dydt
\end{align*}
Differentiating with respect to $T$:
\begin{align*}
\int_{\mathbb{R}}\chi(y)\frac{\partial p(y,T)}{\partial T}dy&=\int_{\mathbb{R}}\bigg(\sigma^2\sum_{k\geq 2}\frac{\epsilon^{(2k-2)}}{(2k)!}\frac{\partial^{2k}\chi}{\partial x^{2k}}\\
&+\sigma^2\big(||\psi_b||^2-||\psi_0||^2\big)\sum_{k\geq 2}\frac{(-\epsilon)^{(2k-3)}}{(2k-1)!}\frac{\partial^{(2k-1)}\chi}{\partial x^{(2k-1)}}\bigg)p(T,y)dy
\end{align*}
If we truncate the right hand side at $k=N$, then the result follows by integrating by parts $2N$ times. Then taking the limit $N\rightarrow\infty$, gives us equation \ref{NG_FP}.
\end{proof}
\subsection{Moment Generating Function:}
Taking the Fourier transform of \ref{NG_FP}, and applying \cite{Dettman} Theorem 8.4.4, we get the following result:
\begin{proposition}\label{moments_prop}
The kth moment of the probability density function described by equation \ref{NG_FP} is given by:
\begin{align}\label{moments}
\mu_k&=\sum_n\frac{k!}{\#OP^n(k)!}\prod_{j\in OP^n(k)}\sigma^2ta_j\\
OP^n(k)&=\text{nth ordered partition of k, not including 1}\nonumber\\
\#OP^n(k)&=\text{number of elements in the nth ordered partition, not including 1}\nonumber
\end{align}
In particular we have for the skew \& kurtosis:
\begin{itemize}
\item $\mu_3=\sigma^2t\epsilon(||\psi_b||^2-||\psi_o||^2)$
\item $\mu_4=3(\sigma^2t)^2+\sigma^2t\epsilon^2$
\end{itemize}
\end{proposition}
\begin{proof}
Taking the Fourier transform of \ref{NG_FP}, we get:
\begin{align*}
\frac{\mathcal{F}(p)}{\partial t}&=\sum_{k\geq 2}a_k(iz)^k\mathcal{F}(p)\\
a_k&=\begin{cases} \frac{\sigma^2\epsilon^{k-2}}{k!}\text{, for k even}\\
\frac{\sigma^2(||\psi_b||^2-||\psi_o||^2)\epsilon^{k-2}}{k!}\text{, for k odd}\end{cases}
\end{align*}
Thus, we get:
\begin{align*}
\mathcal{F}(p)=exp\Big(t\sum_{k\geq 2}a_k(iz)^k\Big)
\end{align*}
Therefore, we can write the Moment Generating function for $p$ as:
\begin{align*}
M_p(z)&=exp\Big(t\sum_{k\geq 2}a_kz^k\Big)\\
&=1+\sum_{k\geq 2}z^k\sum_n\prod_{j\in OP^n(k)}\frac{a_jt}{\#OP^n(k)!}\\
OP^n(k)&=\text{nth ordered partition of k, not including 1}\\
\#OP^n(k)&=\text{number of elements in the nth ordered partition, not including 1}
\end{align*}
\end{proof}
For example, the only partition of 3, not including 1, is $\{3\}$. Therefore, the coefficient of $z^3$ in \ref{moments} is given by:
\begin{align*}
a_3tz^3&=\frac{\sigma^2t\epsilon(||\psi_b||^2-||\psi_o||^2)z^3}{6}
\end{align*}
For 4, we have 2 ordered partitions that do not include 1: $\{4\}$ and $\{2,2\}$. Therefore, the coefficient of $z^4$ in \ref{moments} is given by:
\begin{align*}
\frac{(a_2t)^2}{2!}z^4+a_4z^4&=\frac{(\sigma^2t)^2z^4}{8}+\frac{(\sigma^2t)\epsilon^2z^4}{24}
\end{align*}
The result then follows from the definition of the Moment generating function.
\subsection{Financial Interpretation of the Non-Gaussian Quantum Approach:}
By applying the rotation $S=R(\pi/2)$ in the unitary operator that introduces the stochastic noise, we are essentially representing the case that a buyer that makes up the market state: $\psi_b$, decides to pay the best offer price. Alternatively a seller that makes the market state: $\psi_o$ decides to sell at the best bid price. Indeed in order for a transaction to occur, a buyer \& seller must meet in the middle.

Thus we have via the application: $X\rightarrow S^*XS$, sellers being willing to sell at the lower price currently bid, and buyers willing to pay the higher price currently offered:
\begin{align*}
\begin{pmatrix} x+\epsilon/2 & 0\\0&x-\epsilon/2\end{pmatrix}&\rightarrow\begin{pmatrix} x-\epsilon/2 & 0\\0&x+\epsilon/2\end{pmatrix}
\end{align*}
In terms of the stochastic process: equation \ref{dX_S}, in addition to the usual Gaussian evolution of the mid-price:
\begin{align*}
dj_t(X) &=\begin{pmatrix}0&i\sigma\\-i\sigma&0\end{pmatrix}dA_t+\begin{pmatrix}0&i\sigma\\-i\sigma&0\end{pmatrix}dA^{\dagger}_t
\end{align*}
We allow for a small number of market participants to cross the bid-offer spread via the additional terms:
\begin{align*}
+j_t\begin{pmatrix}-\epsilon & 0\\0&\epsilon\end{pmatrix}d\Lambda_t
\end{align*}
From the form of the associated Fokker-Planck equation: \ref{NG_FP} and associated moments: equation \ref{moments}, we first note that where the bid-offer spread: $\epsilon=0$, the statistical moments reduce to those from the standard Gaussian model. In this case, the standard Fokker-Planck equation is a good approximation. Under these financial circumstances, the market liquidity is such that there is a single unique price at which both buyers \& sellers can transact.

As the bid-offer spread, $\epsilon$, grows so does the excess kurtosis given by: $(\sigma^2t)\epsilon^2$. Note that as one looks further and further into the future, the kurtosis will tend to the Gaussian, in the sense that we have:
\begin{align*}
\lim_{t\rightarrow\infty} \frac{(\sigma^2t)\epsilon^2+3(\sigma^2t)^2}{3(\sigma^2t)^2}&=1
\end{align*}
Note also, that where there are an even balance of buyers \& sellers, we have $||\psi_o||^2=||\psi_b||^2$, and the skew is zero. In this case, non-zero negative skew is introduced by an excess of sellers: $||\psi_o||^2>||\psi_b||^2$, in combination with a non-zero bid-offer spread: $\epsilon$.

%
%
\end{document}